\documentclass[a4paper,authorcolumns,UKenglish,cleveref, autoref,thm-restate]{lipics-v2021}

\usepackage{amsthm,amssymb,amsmath}  
\usepackage{xspace,enumerate}
\usepackage[utf8]{inputenc}
\usepackage{todonotes,url}
\usepackage{amsmath, amsthm, amssymb}
\usepackage[draft,mode=multiuser]{fixme}
\usepackage{thmtools}
\usepackage{thm-restate}

\newenvironment{packed_enum}{
\begin{description}
	\setlength{\topsep}{0pt}
	\setlength{\partopsep}{0pt}
  \setlength{\itemsep}{0pt}
  \setlength{\parskip}{0pt}
  \setlength{\parsep}{0pt}
}{\end{description}}
\newcommand{\Order}{\mathcal{O}}
\newcommand{\TOrder}{\tilde{\mathcal{O}}}

\def\ETFS{\textbf{ETFS}}
\def\AETFS{\textbf{AETFS}}
\def\ETFSDP{\textsc{ETFS-DP}}

\def\dd{\mathinner{.\,.}}

\newtheorem{problem2}{Problem}{\bfseries}{\itshape}

\title{String Sanitization Under Edit Distance: Improved and Generalized} 

\titlerunning{String Sanitization Under Edit Distance: Improved and Generalized} 

\author{Takuya Mieno}
{Kyushu University, Japan \and Japan Society for the Promotion of Science, Japan}
{takuya.mieno@inf.kyushu-u.ac.jp}
{https://orcid.org/0000-0003-2922-9434}
{JSPS KAKENHI Grant Number JP20J11983}
\author{Solon P. Pissis}{CWI, Amsterdam, The Netherlands \and Vrije Universiteit, Amsterdam, The Netherlands}{solon.pissis@cwi.nl}{https://orcid.org/0000-0002-1445-1932}{}
\author{Leen Stougie}
{CWI, Amsterdam, The Netherlands \and Vrije Universiteit, Amsterdam, The Netherlands}
{leen.stougie@cwi.nl}
{}
{Netherlands Organisation for Scientific Research (NWO) through Gravitation-grant NETWORKS-024.002.003}
\author{Michelle Sweering}
{CWI, Amsterdam, The Netherlands}
{michelle.sweering@cwi.nl}
{}
{Netherlands Organisation for Scientific Research (NWO) through Gravitation-grant NETWORKS-024.002.003}

\authorrunning{T.~Mieno et al.} 

\Copyright{Takuya Mieno, Solon P.~Pissis, Leen Stougie, and Michelle Sweering} 



\category{}

\relatedversion{}

\supplement{}

\acknowledgements{}

\nolinenumbers 

\hideLIPIcs  

\EventEditors{Pawe{\l} Gawrychowski and Tatiana Starikovskaya}
\EventNoEds{2}
\EventLongTitle{32nd Annual Symposium on Combinatorial Pattern Matching (CPM 2021)}
\EventShortTitle{CPM 2021}
\EventAcronym{CPM}
\EventYear{2021}
\EventDate{July 5--7, 2021}
\EventLocation{Wroc{\l}aw, Poland}
\EventLogo{}
\SeriesVolume{191}
\ArticleNo{5}

\begin{document}

\maketitle

\begin{abstract}
Let $W$ be a string of length $n$ over an alphabet $\Sigma$, $k$ be a positive integer, and $\mathcal{S}$ be a set of length-$k$ substrings of $W$. The {\ETFS} problem (\textbf{E}dit distance, \textbf{T}otal order, \textbf{F}requency, \textbf{S}anitization) asks us to construct a string  $X_{\mathrm{ED}}$ such that: (i) no string of $\mathcal{S}$ occurs in $X_{\mathrm{ED}}$; (ii) the order of all other length-$k$ substrings over $\Sigma$ (and thus the frequency) is the same in $W$ and in $X_{\mathrm{ED}}$; and (iii) $X_{\mathrm{ED}}$ has minimal edit distance to $W$. When $W$ represents an individual's data and $\mathcal{S}$ represents a set of confidential patterns, the {\ETFS} problem asks for transforming $W$ to preserve its privacy and its utility [Bernardini et al., ECML PKDD 2019].

{\ETFS} can be solved in $\Order(n^2k)$ time [Bernardini et al., CPM 2020]. The same paper shows that {\ETFS} cannot be solved in $\Order(n^{2-\delta})$ time, for any $\delta>0$, unless the Strong Exponential Time Hypothesis (SETH) is false. Our main results can be summarized as follows:
\begin{itemize}
    \item An $\Order(n^2\log^2k)$-time algorithm to solve {\ETFS}.
    \item An $\Order(n^2\log^2n)$-time algorithm to solve {\AETFS} (\textbf{A}rbitrary lengths, \textbf{E}dit distance, \textbf{T}otal order, \textbf{F}requency, \textbf{S}anitization), a generalization of {\ETFS} in which the elements of $\mathcal{S}$ can have arbitrary lengths.
\end{itemize}
Our algorithms are thus optimal up to subpolynomial factors, unless SETH fails. 

In order to arrive at these results, we develop new techniques for computing a variant of the standard dynamic programming (DP) table for edit distance. In particular, we simulate the DP table computation using a directed acyclic graph in which every node is assigned to a smaller DP table. We then focus on redundancy in these DP tables and exploit a tabulation technique according to dyadic intervals to obtain an optimal alignment in $\TOrder(n^2)$ total time\footnote{The notation $\TOrder(f)$ denotes $\Order(f\cdot \text{polylog}(f))$.}. 
Beyond string sanitization, our techniques may inspire solutions to other problems related to regular expressions or context-free grammars.
\end{abstract}

\section{Introduction}\label{sec:intro}

Let us start with an example to ensure that the reader is familiar with the basic motivation behind the computational problem investigated here. Consider a sequence $W$ of items representing a user's on-line purchasing history. Further consider a fragment (or a subsequence) of $W$ denoting that the user has first purchased unscented lotions and zinc/magnesium supplements and then unscented soaps and cotton balls in extra-large bags. By having access to $W$ and to the respective domain knowledge, one can infer that the user is probably pregnant and close to the delivery date. 

\emph{Data sanitization}, also known as \emph{knowledge hiding}, is a privacy-preserving data mining process aiming to prevent the mining of confidential knowledge from published datasets; it has been an active area of research for the past 25 years~\cite{CM96,DBLP:conf/icdm/OliveiraZ03,DBLP:journals/tkde/VerykiosEBSD04,DBLP:conf/cikm/Gkoulalas-DivanisV06,DBLP:journals/tkde/WuCC07,DBLP:journals/tkde/Gkoulalas-DivanisV09,DBLP:journals/tkde/AbulBG10,DBLP:journals/dke/AbulG12,DBLP:conf/kdd/Gkoulalas-DivanisL11,DBLP:conf/icdm/GwaderaGL13,DBLP:conf/sdm/LoukidesG15,DBLP:conf/wsdm/BonomiFJ16,DBLP:conf/pkdd/0001CCGLPPR19,TKDD,DBLP:conf/icdm/0001CGGLPPPSS20}.
Informally, it is the process of disguising (hiding) confidential information in a given dataset. This process typically incurs some data utility loss that should be minimized. Thus, naturally, privacy constraints and utility objective functions lead to the formulation of combinatorial optimization problems. From a fundamental perspective, it is thus relevant to be able to establish some formal guarantees.

A string $W$ is a sequence of letters over some alphabet $\Sigma$. 
Individuals' data, in domains ranging from web analytics to transportation and bioinformatics, are typically represented by strings. For example, when $\Sigma$ is a set of items, 
$W$ can represent a user's purchasing history~\cite{DBLP:conf/icde/AgrawalS95}; 
when $\Sigma$ is a set of locations, $W$ can represent a
user's location profile~\cite{DBLP:conf/gis/YingLWT11}; 
and when $\Sigma$ is the DNA alphabet, $W$ can represent
a patient's genome sequence~\cite{NGS}. Such strings commonly fuel up
a gamut of applications; in particular, frequent pattern mining applications~\cite{DBLP:books/sp/fpm14/Aggarwal14a}. For example, frequent pattern mining from location history data facilitates route planning~\cite{DBLP:journals/is/ChenYL15}; frequent pattern mining from market-basket data facilitates business decision making~\cite{DBLP:conf/icde/AgrawalS95}; frequent pattern mining from genome sequences facilitates clinical diagnostics~\cite{NGS}. 
To support these applications in a privacy-preserving manner, individual sequences are often being disseminated after they have been sanitized.

Towards this end, Bernardini et al.~have recently formalized the following string sanitization problem under edit distance~\cite{TKDD}.
Let $W$ be a string of length $n$ over an alphabet $\Sigma$, $k$ be a positive integer, and $\mathcal{S}$ be a set of length-$k$ substrings of $W$. Set $\mathcal{S}$ is conceptually seen as an {\em antidictionary}: a set of {\em sensitive patterns} modelling private or confidential information. The {\ETFS} problem (\textbf{E}dit distance, \textbf{T}otal order, \textbf{F}requency, \textbf{S}anitization) asks us to construct a string  $X_{\mathrm{ED}}$ such that: (i) no string of $\mathcal{S}$ occurs in $X_{\mathrm{ED}}$; (ii) the order of all other length-$k$ substrings over $\Sigma$ is the same in $W$ and in $X_{\mathrm{ED}}$; and (iii) $X_{\mathrm{ED}}$ has minimal edit distance to $W$. In order to obtain a feasible solution string, we may need to extend $\Sigma$ to $\Sigma_{\#}=\Sigma\cup\{\#\} $, which includes a special letter $\#\notin \Sigma$.

\begin{example} \label{ex:ETFS}
Let $W=\texttt{ecabaaaaabbbadf}$ over alphabet $\Sigma=\{\texttt{a},\texttt{b},\texttt{c},\texttt{d},\texttt{e},\texttt{f}\}$ be the input string. Further let $k = 3$ and the set of sensitive patterns be $\mathcal{S}=\{\texttt{aba},\texttt{baa},\texttt{aaa},\texttt{aab},\texttt{bba}\}$. Consider the following three feasible (sanitized) strings: $X_{\mathrm{TR}} = \texttt{eca\#cab\#abb\#bbb\#bad\#adf}$, $X_{\mathrm{MIN}} = \texttt{ecabbb\#badf}$ and $X_{\mathrm{ED}} = \texttt{ecab\#aa\#abbb\#badf}$. All three strings contain {\em no sensitive pattern} and preserve the {\em total order} and thus the {\em frequency} of all non-sensitive length-$3$ patterns of $W$: $X_{\mathrm{TR}}$ is the trivial solution of interleaving the non-sensitive length-$3$ patterns of $W$ with $\#$;  $X_{\mathrm{MIN}}$ is the {\em shortest} possible such string; and $X_{\mathrm{ED}}$ is a string {\em closest} to $W$ in terms of edit distance. 
\end{example}

A simple $\Order(n^2k|\Sigma|)$-time solution~\cite{TKDD} to {\ETFS} can be obtained via employing approximate regular expression matching. Consider the regular expression $R$ that encodes all feasible solution strings. The size of $R$ is $\Order(nk|\Sigma|)$. By aligning $W$ and $R$ using the standard quadratic-time algorithm~\cite{regex}, we obtain an optimal solution $X_{\mathrm{ED}}$ in $\Order(n^2k|\Sigma|)$ time for {\ETFS}. Bernardini et al.~showed that this can be improved to $\Order(n^2k)$ time~\cite{bernardini_et_al:LIPIcs:2020:12132}. Let us informally describe their algorithm. (A formal description of their algorithm follows in Section~\ref{sec:bg}.) We use a dynamic programming (DP) table similar to the standard edit distance algorithm. We write the letters of the input string $W$ on the top of the first row. Since we do not know the exact form of the output string $X_{\mathrm{ED}}$, we write the non-sensitive length-$k$ patterns to the left of the first column interleaved by special $\#$ letters. We then proceed to fill this table using recursive formulae. The formulae are more involved than the edit distance ones to account for the possibility to {\em merge} consecutive non-sensitive patterns (e.g., \texttt{eca} and \texttt{cab} are merged to \texttt{ecab} in Example~\ref{ex:ETFS}) and to expand the $\#$'s into longer gadgets that may contain up to $k-1$ letters from $\Sigma$ (e.g., \texttt{\#aa\#} in Example~\ref{ex:ETFS}). Once the DP table is filled, we construct an $X_{\mathrm{ED}}$ by tracing back an optimal alignment.

Bernardini et al.~also showed, via a reduction from the weighted edit distance problem~\cite{DBLP:conf/focs/BringmannK15}, that {\ETFS} cannot be solved in $\Order(n^{2-\delta})$ time, for any $\delta>0$, unless the strong exponential time hypothesis (SETH)~\cite{DBLP:journals/jcss/ImpagliazzoP01,DBLP:journals/jcss/ImpagliazzoPZ01} is false. We were thus also motivated to match this lower bound.

\paragraph*{Our Results and Techniques} Our first main result is the following.

\begin{restatable}{theorem}{etfsatheorem}\label{the:etfs}
 The {\ETFS} problem can be solved in $\Order(n^2\log^2k)$ time.
\end{restatable}

Let us also stress that the algorithm underlying Theorem~\ref{the:etfs} works under edit distance with arbitrary weights at no extra cost.

We also consider a generalized version of {\ETFS}, which we denote by {\AETFS} (\textbf{A}rbitrary lengths, \textbf{E}dit distance, \textbf{T}otal order, \textbf{F}requency, \textbf{S}anitization). The only difference in {\AETFS} with respect to {\ETFS} is that $\mathcal{S}$ can contain elements (sensitive patterns) of arbitrary lengths. This generalization is evidently more useful as it drops the restriction of fixed-length sensitive patterns; it also turns out to be algorithmically much more challenging. In both {\ETFS} and {\AETFS}, we make the standard assumption that substrings of $W$ are represented as intervals over $[0,n-1]$, and thus each element in $\mathcal{S}$ has an $\Order(1)$-sized representation. We further assume that $\mathcal{S}$ satisfies the properties of \emph{closure} and \emph{minimality} (formally defined in Section~\ref{sec:bg}), which in turn ensure that $\mathcal{S}$ has an $\Order(n)$-sized representation.
\begin{example} \label{ex:AETFS}
  Consider the same input string $W=\texttt{ecabaaaaabbbadf}$ as in Example~\ref{ex:ETFS}. Further let $k = 3$ and the set of sensitive patterns be $\mathcal{S}=\{\texttt{aba},\texttt{aa},\texttt{abbba}\}$.
  Then, string $Y_{\mathrm{ED}} = \texttt{ecab\#abb\#bbbadf}$ is a feasible string and is a closest to $W$ in terms of edit distance.
  Notice that, we cannot merge all of the three consecutive non-sensitive patterns $\texttt{abb}$, $\texttt{bbb}$, and $\texttt{bba}$ into one since it will result in an occurrence of the sensitive pattern \texttt{abbba}; we thus rather create \texttt{abb\#bbba}.
\end{example}

Our second main result is the following.

\begin{restatable}{theorem}{aetfsatheorem}\label{the:aetfs}
 The {\AETFS} problem can be solved in $\Order(n^2\log^2n)$ time.
\end{restatable}

Our algorithms are thus optimal up to subpolynomial factors, unless SETH fails. 
Let us describe the main ideas behind the new techniques we develop. As in Example~\ref{ex:AETFS}, a sensitive pattern of length greater than $k$ might be generated by merging multiple non-sensitive patterns. In {\AETFS}, we have to consider avoiding such \emph{invalid} merge operations. 
If we enumerate all \emph{valid} combinations of merging non-sensitive patterns, and run the DP for {\ETFS} for all the cases, then we can obtain an optimal solution to {\AETFS}. Our main idea for reducing the time complexity is to carefully maintain a directed acyclic graph~(DAG) for representing all such valid combinations. We first construct the DAG, and then plug a small DP table into each node of the DAG. This technique gives us an $\Order(n^3)$-time solution to {\AETFS}. To achieve $\TOrder(n^2)$ time, we focus on redundancy in the DP tables. When the size of the DP tables is large, there must be multiple sub-tables corresponding to the same pair of strings. Before propagating, we precompute lookup table structures of size $\Order(n^2\log^2n)$ according to dyadic intervals on $[0,n-1]$. To this end, we modify the data structure proposed in~\cite{brubach2018succinct}. Then, we decompose the DP tables into sub-tables according to these dyadic intervals. We compute only boundaries of such sub-tables using the precomputed lookup table structures, and thus, we obtain an optimal alignment for {\AETFS} in $\Order(n^2\log^2n)$ total time. By applying the same technique to {\ETFS}, we obtain an $\Order(n^2\log^2k)$-time solution, which improves the state of the art by a factor of $k/\log^2k$~\cite{bernardini_et_al:LIPIcs:2020:12132}. 

In a nutshell, our main technical contribution is that we manage to align a string of length $n$ and a specific regular expression of size $\Omega(nk|\Sigma|)$ in $\TOrder(n^2)$ time. We can also view the solution spaces of {\ETFS} and {\AETFS} as context-free languages. The main idea of our {\AETFS} algorithm is to first preprocess a set $N$ of non-terminals, such that we can later use them in $\Order(n)$ time each. We then write the context-free language as a new language, which is accepted by a Deterministic Acyclic Finite State Automaton (DASFA), taking the set $N$ as its terminals. In this paper, we develop several techniques to reduce the size of the DAFSA (cf.~DAG) to $\tilde{\Order}(n)$ and efficiently precompute the set $N$ (cf.~lookup tables) in $\tilde{\Order}(n^2)$ time. Thus, beyond string sanitization, our techniques may inspire solutions to other problems related to regular expressions or context-free grammars.

\paragraph*{Paper Organization} Section~\ref{sec:bg} introduces the basic definitions and notation used throughout, and also provides a summary of the currently fastest algorithm for solving {\ETFS}~\cite{bernardini_et_al:LIPIcs:2020:12132}. In Section~\ref{sec:tab}, we describe our lookup table structures. In Section~\ref{sec:AETFS}, we present the $\Order(n^3)$-time algorithm for solving {\AETFS}. This algorithm is refined to an $\TOrder(n^2)$-time algorithm, which is described in Section~\ref{sec:nsquared}. Along the way, in Section~\ref{sec:nsquared}, we also describe an $\TOrder(n^2)$-time algorithm for solving {\ETFS}. 


\section{Preliminaries}\label{sec:bg}
\paragraph*{Strings} 
An \emph{alphabet} $\Sigma$ is a finite set of elements called $\emph{letters}$.
Let $S=S[0]S[1]\cdots S[n-1]$ be a \emph{string} of length $|S|=n$ over an alphabet $\Sigma$ of size $\sigma=|\Sigma|$. Let $\Gamma = \{\ominus, \oplus, \otimes\}$ be a set of special letters with $\Gamma\cap\Sigma = \emptyset$.
By $\Sigma^*$ we denote the set of all strings over $\Sigma$, and by $\Sigma^k$ the set of all length-$k$ strings over $\Sigma$.
For two indices $0 \leq i \leq j \leq n-1$,  $S[i\dd j]=S[i]\cdots S[j]$ is the \emph{substring} of $S$ that starts at position $i$ and ends at position $j$ of $S$. By $\varepsilon$ we denote the \emph{empty string} of length $0$. A \emph{prefix} of $S$ is a substring of the form $S[0\dd j]$, and a \emph{suffix} of $S$ is a substring of the form $S[i\dd n-1]$. Given two strings $U$ and $V$ we say that $U$ has a {\em suffix-prefix overlap} of length $\ell>0$ with $V$ if and only if the length-$\ell$ suffix of $U$ is equal to the length-$\ell$ prefix of $V$, i.e., $U[|U|-\ell \dd |U|-1]=V[0\dd \ell-1]$.

We fix a string $W$ of length $n$ over an alphabet $\Sigma$. We assume that $\Sigma=\{1,\ldots,n^{\Order(1)}\}$. If this is not the case, we use perfect hashing~\cite{DBLP:journals/jacm/FredmanKS84} to hash $W[i]$, for all $i\in[1,n]$, and obtain another string over $\Sigma=\{1,\ldots,n\}$ in $\Order(n)$ time with high probability or in $\Order(n \log n)$ time deterministically via sorting. We consider the obtained string to be $W$. We also fix an integer $0<k<n$.
Unless specified otherwise, we refer to a length-$k$ string or a {\em pattern} interchangeably. An occurrence of a pattern is uniquely defined by its starting position. Let $\mathcal{S}_k$ be the set representing the sensitive patterns as starting positions over $\{0,\ldots, n-k\}$ with the following closure property: for every $i\in\mathcal{S}_k$, any $j$ for which $W[j\dd j+k-1]=W[i\dd i+k-1]$, must also belong to $\mathcal{S}_k$. That is, if an occurrence of a pattern is in $\mathcal{S}_k$, then all its occurrences are in $\mathcal{S}_k$. A substring $W[i\dd i+k-1]$ of $W$ is called {\em sensitive} if and only if $i\in \mathcal{S}_k$; $\mathcal{S}_k$ is thus the complete set of occurrences of sensitive patterns. The difference set $\mathcal{I}=\{0,\ldots, n-k\} \setminus \mathcal{S}_k$ is the set of occurrences of length-$k$ {\em non-sensitive} patterns. 

For any substring $U$, we denote by $\mathcal{I}_U$ the set of occurrences in $U$ of non-sensitive length-$k$ strings over $\Sigma$. (We have that $\mathcal{I}_W=\mathcal{I}$.) We call an occurrence $i$ the {\em t-predecessor} of another occurrence $j$ in $\mathcal{I}_U$ if and only if $i$ is the largest element in $\mathcal{I}_U$  that is less than $j$. This relation induces a {\em strict total order} on the occurrences in $\mathcal{I}_U$. We call a subset $\mathcal{J}$ of $\mathcal{I}_U$ a \textit{t-chain} if for all elements in $\mathcal{J}$ except the minimum one, their t-predecessor is also in $\mathcal{J}$. For two strings $U$ and $V$, chains $\mathcal{J}_U$ and $\mathcal{J}_V$ are {\em equivalent}, denoted by $\mathcal{J}_U \equiv \mathcal{J}_V$, if and only if $|\mathcal{J}_U|=|\mathcal{J}_V|$ and $U[u\dd u+k-1]=V[v\dd v+k-1]$, where $u$ is the $j$-th smallest element of $\mathcal{J}_U$ and $v$ is the $j$-th smallest of $\mathcal{J}_V$, for all $j\le |\mathcal{J}_U|$.

Given two strings $U$ and $V$ the {\em edit distance} $d_\mathrm{E}(U,V)$ is defined as the minimum number of elementary edit operations (letter insertion, deletion, or substitution) that transform one string into the other. Each edit operation can also be associated with a cost: a fixed positive value. Given two strings $U$ and $V$ the {\em weighted edit distance} $d_{\mathrm{WE}}(U,V)$ is defined as the minimal total cost of a sequence of edit operations to transform one string into the other. We assume throughout that the three edit operations all have unit weight. However, as mentioned in Section~\ref{sec:intro}, our algorithm for the \ETFS{} problem (formally defined below) also works for arbitrary weights at no extra cost. The standard algorithm to compute the edit distance between two strings $U$ and $V$~\cite{Lev} works by creating a $(|U|+1) \times (|V|+1)$ DP table $D$ with $D[i][j] = d_{\mathrm{E}}(U[0\dd i-1], V[0\dd j-1])$. The sought edit distance is thus $d_{\mathrm{E}}(U, V) = D[|U|][|V|]$. Since we compute each table entry from the entries to the left, top and top-left in $\Order(1)$ time, the algorithm runs in $\Order(|U|\cdot|V|)$ time. Moreover, we can find an optimal (minimum cost) alignment by tracing back through the table.

\paragraph*{The {\ETFS} Problem} We formally define {\ETFS}, one of the problems considered in this paper.

\begin{problem2}[{\ETFS}]\label{prob:et-string} 
Given a string $W$ of length $n$, an integer $k>1$, and a set $\mathcal{S}_k$ (and thus set $\mathcal{I}$), construct a string $X_{\mathrm{ED}}$ which is at minimal (weighted) edit distance from $W$ and satisfies:
\begin{packed_enum}
    \item[\textbf{C1}] $X_{\mathrm{ED}}$ does not contain any sensitive pattern.
    \item[\textbf{P1}] $\mathcal{I}_W \equiv \mathcal{I}_{X_{\mathrm{ED}}}$, i.e., the t-chains $\mathcal{I}_W$ and $\mathcal{I}_{X_{\mathrm{ED}}}$ are equivalent.
\end{packed_enum}
\end{problem2}

\paragraph*{The {\AETFS} Problem}
The length of sensitive patterns in the {\ETFS} setting is fixed.
In what follows, we define a generalization of the {\ETFS} problem which allows for arbitrary length sensitive patterns. Let $\mathcal{S}$ be a set of intervals with the two following properties~(closure property and minimality property):~(i) For every $[i, j]\in\mathcal{S}$, any $[i', j']$ for which $W[i'\dd j']=W[i\dd j]$, must also belong to $\mathcal{S}$; and~(ii) any proper sub-interval of $[i, j] \in \mathcal{S}$ is not in $\mathcal{S}$.
It is easy to see that $|\mathcal{S}| \le n$ from its minimality.
Now, we redefine notions of sensitive and non-sensitive patterns as follows:
A sensitive pattern is an \emph{arbitrary length} substring $W[i\dd j]$ of $W$ for each $[i, j] \in \mathcal{S}$.
For a fixed $k$, a non-sensitive pattern is a length-$k$ substring of $W$ containing no sensitive pattern as a substring.
\begin{problem2}[{\AETFS}]\label{prob:aet-string} 
Given a string $W$ of length $n$, an integer $k>1$, and a set $\mathcal{S}$ (and thus set $\mathcal{I}$), construct a string $Y_{\mathrm{ED}}$ which is at minimal edit distance from $W$ and satisfies:
\begin{packed_enum}
    \item[\textbf{C1}] $Y_{\mathrm{ED}}$ does not contain any sensitive pattern.
    \item[\textbf{P1}] $\mathcal{I}_W \equiv \mathcal{I}_{Y_{\mathrm{ED}}}$, i.e., the t-chains $\mathcal{I}_W$ and $\mathcal{I}_{Y_{\mathrm{ED}}}$ are equivalent.
\end{packed_enum}
\end{problem2}

\paragraph*{The {\ETFSDP} Algorithm}For independent reading we describe here {\ETFSDP}, the algorithm from~\cite{bernardini_et_al:LIPIcs:2020:12132} that solves the {\ETFS} problem in $\Order(n^2k)$ time. The output string $X_{\mathrm{ED}}$ is a string that contains all non-sensitive patterns in the same order as in $W$.
For each pair of consecutive non-sensitive patterns, their occurrences in $X_{\mathrm{ED}}$ are either~(i) overlapping by $k-1$ letters (e.g., \texttt{eca} and \texttt{cab} in Example~\ref{ex:ETFS}) or~(ii) delimited by a string over $\Sigma_\#$ which contains no length-$k$ string over $\Sigma$~(e.g., $\texttt{\#aa\#}$ in Example~\ref{ex:ETFS}). We call such strings \emph{gadgets}. For case~(ii), we use the following regular expressions:
$$
\Sigma^{<k} = (a_1|a_2|\ldots|a_{|\Sigma|}|\varepsilon)^{k-1},
$$
where $\Sigma = \{a_1, a_2, \ldots, a_{|\Sigma|}\}$.
Also, the special letters $\ominus, \oplus, \otimes \in \Gamma$ correspond to regular expressions $(\Sigma^{<k}\#)^\ast$, $\#(\Sigma^{<k}\#)^\ast$, and $(\#\Sigma^{<k})^\ast$, respectively.
Let $N_0, N_1,\ldots, N_{|\mathcal{I}|-1}$ be the sequence of non-sensitive patterns sorted in the order in which they occur in $W$. In what follows, we fix string $T = \ominus N_0\oplus N_1\oplus\cdots\oplus N_{|\mathcal{I}|-1}\otimes$ of length $(k+1)|\mathcal{I}|+1$. String $T$ corresponds to the regular expression $R$ that represents the set of all feasible solutions (feasible strings) in which all non-sensitive patterns in the string are delimited by stings over $\Sigma_\#$.
Moreover, we need to consider feasible strings in which a non-sensitive pattern overlaps the next one. Let $M$ be a binary array of length $|\mathcal{I}|$ such that for each $0 \le i \le |\mathcal{I}|-1$, $M[i] = 1$ if $i > 0$ and $N_{i-1}$ has a suffix-prefix overlap of length $k-1$ with $N_i$, and $M[i] = 0$ otherwise. Namely, $M[i] = 1$ implies that $N_{i-1}$ and $N_i$ can be merged for $0 < i \le |\mathcal{I}|-1$.

Let $E$ be a table of size $((k+1)|\mathcal{I}|+1)\times(n+1)$.
The rows of $E$ correspond to string $T$ defined above and the columns to string $W$.
Note that the leftmost column corresponds to the empty string $\varepsilon$ as in the standard edit distance DP table.
Each cell $E[i][j]$ contains the edit distance between the regular expression corresponding to $T[0\dd i]$ and $W[0\dd j-1]$.
We classify the rows of $E$ into three categories: {\em gadget rows}; {\em possibly mergeable rows}; and {\em ordinary rows}.
We call every row corresponding to a special letter in $\Gamma$ a gadget row. Namely, rows with index $i \equiv 0\bmod (k+1)$ are gadget rows.
Also, we call every row corresponding to the last letter of a non-sensitive pattern a possibly mergeable row. Namely, rows with index $i \equiv -1\bmod (k+1)$ are possibly mergeable rows.
All the other rows are called ordinary rows.
The recursive formula of ordinary rows is the same as in the standard edit distance solution:
\begin{equation*}\label{eq:ETFSordrow}
E[i][j] = \min\left. \begin{cases}E[i-1][j] + 1, &\text{(insert}\text{)}\\ E[i][j-1] + 1, &\text{(delete}\text{)}\\E[i-1][j-1] + I[T[i] \neq W[j-1]], &\text{(match or substitute),}\end{cases} \right.
\end{equation*}
where $I$ is an indicator function: $I[T[i] \neq W[j-1]] = 1$ if $T[i]\neq W[j-1]$, and 0 otherwise.
Next, consider a possibly mergeable row $E[i][\cdot]$ which is the last row of the non-sensitive pattern $N_h$. If $M[h] = 0$, then the recursive formula is the same as that of ordinary rows. Otherwise ($M[h] = 1$), $N_{h-1}$ and $N_h$ can be merged. Merging them means that the values in the previous mergeable row $E[i-k-1][\cdot]$ will be propagated to $E[i][\cdot]$ directly without considering the $k$ rows below. Thus, the recursive formula is:
\begin{equation*}\label{eq:ETFSmergerow}
E[i][j] = \min\left.\begin{cases}
E[i-1][j]+1, &\text{\hspace{1.8cm} (insert)}\\ 
E[i][j-1]+1, &\text{\hspace{1.8cm} (delete)}\\ 
E[i-1][j-1]+I[T[i] \neq W[j-1]], &\text{\hspace{1.8cm} (match or substitute)}\\
E[i-k-1][j]+1, &\text{if}~M[h]=1~\text{(insert and  merge)}\\
E[i-k-1][j-1] + I[T[i]\neq W[j-1]],&\text{if}~M[h]=1~\text{(match or sub.~and merge).}
\end{cases}\right.
\end{equation*}

Next, consider a gadget row $E[i][\cdot]$ which corresponds to a special letter in $\Gamma$.
Because of the form of regular expressions corresponding to special letters, a \# can either be inserted or substituted directly after a non-sensitive pattern, or be preceded by another \# no more than $k$ positions earlier. This results in the following recursive formula:
\begin{equation*}\label{eq:ETFSgadgetrow}
E[i][j] = \min\left. \begin{cases}
E[i-1][j]+1, &\text{(insert)}\\
E[i-1][j-1]+1, &\text{(substitute)}\\
E[i][j-1]+1,\ldots, E[i][\max\{0,j-k\}]+1, &\text{(delete or extend gadget).}
\end{cases}\right.
\end{equation*}

For completeness, we write down the recursive formula for initializing the leftmost column:
\begin{equation*}\label{eq:ETFSbasecol}
E[i][0] = \begin{cases}
E[i-k-1][0]+1, & \text{if}~ i \equiv -1\bmod (k+1)\land M[h]=1~\text{(merge)}\\
E[i-1][0]+1, & \text{otherwise (no merge).}
\end{cases}
\end{equation*}

Unlike in the standard setting~\cite{Lev}, the edit distance between $W$ and any string matching the regular expression $R$ is not necessarily found in its bottom-right entry $E[|\mathcal{I}|(k+1)][|W|]$. 
Instead, it is found among the rightmost $k$ entries of the last row (in case $X_{\mathrm{ED}}$ ends with a string in $\otimes$), and the rightmost entry of the second-last row (when $X_{\mathrm{ED}}$ ends with the last letter of the last non-sensitive pattern). After computing the edit distance value, we construct an $X_{\mathrm{ED}}$. To do so, when computing each entry $E[i][j]$, we memorize a backward-pointer to an entry from which the minimum value for $E[i][j]$ was obtained.
We then construct $X_{\mathrm{ED}}$ from right to left with respect to the sequence of edit operations corresponding to an optimal alignment obtained by the backward-pointers.

\section{Compact Lookup Table Structure for Squared Blocks}\label{sec:tab}
In this section we consider the standard edit distance table, and propose a data structure which can answer some queries on
a $b\times b$ sub-table of the DP table, which we call a 
\emph{block}, corresponding to two strings of the same length.
Our data structure is similar to the one proposed in~\cite{brubach2018succinct}, tailored, however, to our needs. We next provide some further definitions about blocks. Let $B$ be a $b\times b$ block to be processed. The top~(resp.~bottom) row of $B$ is called the input~(resp.~output) row of $B$. Similarly, the leftmost~(resp.~rightmost) column of $B$ is called the input~(resp.~output) column of $B$. A cell in the input~(resp.~output) row or column is called an input~(resp.~output) cell.

In the following, we propose a lookup table for $b\times b$ blocks that computes all output cell values of a block in $\Order(b)$ time for any given block and input cell values of the block. We modify the following known result to enhance it with a trace-back functionality.
%
\begin{theorem}[Theorem~1 in \cite{brubach2018succinct}] \label{thm:brubach}
  Given two strings both of length $b$ corresponding to a $b\times b$ block,
  we can construct a data structure of size $\Order(b^2)$ in $\Order(b^2\log b)$ time
  such that given any values for the input row and column of the block,
  the data structure can compute the output row and column of the block in $\Order(b)$ time.
\end{theorem}
In~\cite{brubach2018succinct}, the authors did not refer to tracing back, i.e., it is not clear how to obtain an optimal alignment using Theorem~\ref{thm:brubach}. We prove that we can trace back a shortest path to an output cell in a $b\times b$ block in $\Order(b)$ time using $\Order(b^2\log b)$ additional space. This yields an optimal alignment.
We next briefly describe the data structure of Theorem~\ref{thm:brubach} and explain how we modify it.

\subsection{Constructing a Data Structure for a Pair of Strings}
For constructing the data structure of Theorem~\ref{thm:brubach},
Brubach and Ghurye~\cite{brubach2018succinct} utilize the result of \cite{schmidt1998all}.
Instead, we use the following result by Klein~\cite{klein2005multiple}, which is more general.

\begin{theorem}[\cite{klein2005multiple}] \label{thm:klein}
  Given an $N$-node planar graph with non-negative edge labels,
  we can construct a data structure of size $\Order(N\log N)$ in $\Order(N\log N)$ time 
  such that given a node $s$ in the graph and 
  another node $t$ on the boundary of the infinite face,
  the data structure can compute the maximum~(or minimum) distance
  from $s$ to $t$ in $\Order(\log N)$ time. 
  Also, if the graph has constant degree,
  then we can compute the shortest $s$-$t$ path in time linear in the length of the path.
\end{theorem}

The lookup table structure for a block is constructed as follows~(see~\cite{brubach2018succinct} for details).
Let $B$ be a $b\times b$ block to be preprocessed.
First, we regard $B$ as a grid-graph of size $b\times b$.
Namely, each node corresponds to a cell in the block, and each edge corresponds to an edit operation. Also, each edge is labeled by the weight of its corresponding edit operation.
Then, we construct the data structure of Theorem~\ref{thm:klein} for the grid-graph.
We denote this data structure by $\mathcal{D}_B$.
Next, for each input cell $u$ and each output cell $v$, we compute the weight of the shortest path from $u$ to $v$, and store them to table $M_B$ of size $(2b-1) \times (2b-1)$.
Each row~(resp. column) of $M_B$ corresponds to each output~(resp. input) cell of $B$.
A table is called {\em monotone} if each row's minimum value occurs in a column which is equal to or greater than the column of the previous row's minimum. It is {\em totally monotone} if the same property is true for every sub-table defined by an arbitrary subset of the rows and columns of the given table. It is known that we can construct $M_B$ so that it is totally monotone~\cite{crochemore2003subquadratic}. 
We thus construct $\mathcal{D}_B$ and $M_B$ in $\Order(b^2\log b)$ time and space.

By Theorem~\ref{thm:brubach}, the size of the final data structure (that depends on the size of $M_B$) is $\Order(b^2)$. However, $\Order(b^2\log b)$ working space is used for constructing $\mathcal{D}_B$.
In our algorithm, we also use table $M_B$ and keep the temporary data structure $\mathcal{D}_B$ to support tracing back operations efficiently.

\subsection{Answering Queries and Tracing Back}
Given a query input row and column, 
we can compute the output row and column in $\Order(b)$ time
using the SMAWK algorithm~\cite{aggarwal1987geometric} for finding the minimum value in each row of an implicitly-defined totally monotone table,
since $M_B$ is totally monotone~\cite{brubach2018succinct}.
Note that, for each output cell $v$ of $B$,
we can also obtain an input cell $s_v$
which is the starting cell of a shortest path ending at $v$
from the result of SMAWK algorithm.
Thus, by using data structure $\mathcal{D}_B$,
we can obtain a shortest $s_v$-$v$ path in time linear in the length of the path.
To summarize, we obtain the following lemma.
\begin{lemma} \label{lem:traceback}
  Given two strings both of length $b$ corresponding to a $b\times b$ block,
  we can construct a data structure of size $\Order(b^2\log b)$ in $\Order(b^2\log b)$ time
  such that given any values for the input row and column of the block,
  the data structure can compute the output row and column of the block in $\Order(b)$ time.
  Furthermore, given an output cell $v$ and any other cell $u$ in the block, we can compute a shortest $u$-$v$ path in time linear in the length of the path.
\end{lemma}
This data structure works under edit distance with arbitrary weights at no extra cost.

\section{Sensitive Patterns of Arbitrary Lengths} \label{sec:AETFS}
In this section we propose a data structure with which we can solve the {\AETFS} problem in time $\Order(n^3)$.
First, let us consider whether {\ETFSDP} can be applied directly to the {\AETFS} problem.
The {\AETFS} problem is a generalization of the {\ETFS} problem, and there are some differences between them: 
if there exists a \emph{long sensitive pattern} of length longer than $k$, then we cannot apply the same logic for the possibly mergeable rows to the {\AETFS} problem.
This is because merging multiple non-sensitive patterns of length $k$ may create a long sensitive pattern, while this sensitive pattern must be hidden.
In contrast, if there exists a \emph{short sensitive pattern} of length less than $k$, then we cannot apply the same logic for the gadget rows to the {\AETFS} problem, since this may introduce a short sensitive pattern in a gadget. Thus {\AETFS} is much more challenging.

Let $L=\Order(n^2)$ denote the total length of long sensitive patterns. As a first step towards our main result, we prove the following lemma.
\begin{lemma}\label{lem:AETFS}
  The {\AETFS} problem can be solved in $\Order(k|\mathcal{I}|n + Ln)$ time.
\end{lemma}

Note that Lemma~\ref{lem:AETFS} yields $\Order(n^2k)$ time for {\ETFS} because in this case $L = 0$. Lemma~\ref{lem:AETFS} thus generalizes Theorem~2 in~\cite{bernardini_et_al:LIPIcs:2020:12132}.
In what follows, we propose a new data structure for solving the {\AETFS} problem and prove Lemma~\ref{lem:AETFS}.
The main idea is to use multiple DP tables and link them under specific rules. Interestingly, our data structure is shaped as a DAG consisting of DP tables.

\subsection{Long Sensitive Patterns} \label{sec:longsensitive}
If there is a long sensitive pattern, we need to consider the case
where multiple non-sensitive patterns are contained in a single sensitive pattern. (Recall that all non-sensitive patterns have fixed length $k$.) In this case, we cannot apply the {\ETFSDP} algorithm from~\cite{bernardini_et_al:LIPIcs:2020:12132} directly.

Let us consider the situation in which we have just finished computing a possibly mergeable row. We may be able to choose the next move from two candidates: either go down to the next (gadget) row
or jump to the next possibly mergeable row if possible.
We consider a \emph{decision tree} $\mathcal{T}$ that represents all combinations of such choices at all possibly mergeable rows~(inspect Figure~\ref{fig:pruning}). We regard $\mathcal{T}$ as a \emph{tree of tables}, i.e., each node of $\mathcal{T}$ represents a small DP table. Let $E[0\dd  (k+1)|\mathcal{I}|][0\dd  n]$ be the DP table of the {\ETFSDP} algorithm described in Section~\ref{sec:bg}. There are three types of nodes in $\mathcal{T}$:
\emph{root}, \emph{\#-node}, and \emph{m-node}.
The root represents sub-table $E[0\dd  k][0\dd n]$.
For each depth $b$ with $1 \le b \le |\mathcal{I}|-1$,
the \#-node at depth $b$ represents sub-table $E[b(k+1)\dd  (b+1)(k+1)-1][0\dd n]$, and each m-node at depth $b$ represents a copy of possibly mergeable row $E[(b+1)(k+1)-1][0\dd n]$.
Each edge~$(u, v)$ of $\mathcal{T}$ means that the bottom row values of $u$ will be propagated to the top row of $v$. If there are multiple incoming edges $(u_1, v), (u_2, v), \ldots, (u_p, v)$ of a single node $v$, then we virtually consider a row $r[0\dd n]$ as the previous possibly mergeable row of $v$ such that $r[j]$ is the minimum value between all $j$-th values in the last rows of $u_1, u_2, \ldots, u_p$ for each $0 \le j \le n$. We call a path that consists of only m-nodes an \emph{m-path}.

\begin{figure}[!t]
\centering
\includegraphics[width=\linewidth]{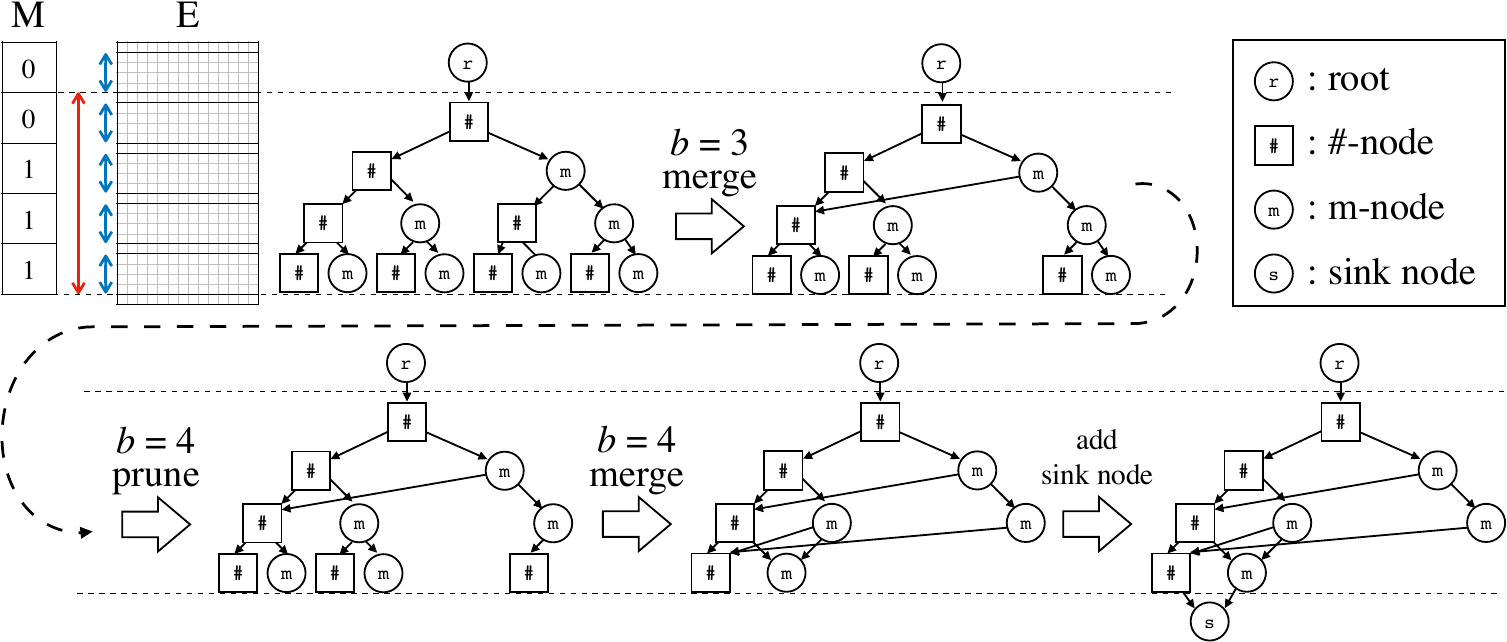}
\caption{
  An example for pruning and merging a decision tree.
  The red arrow represents a long sensitive pattern, and blue arrows represent non-sensitive patterns.
  The tree on the top-left represents the decision tree after pruning or merging operations for depths $b = 0, 1, 2$.
  For $b = 3$, we merge two \#-nodes by Rule 3.
  For $b = 4$, we prune an m-node by Rule 2 since merging four non-sensitive patterns following the node results creating a sensitive pattern.
  Furthermore, we merge three \#-nodes by Rule 3, and merge two m-nodes by Rule 4.
  Finally, we add the sink node at the bottom.
}
\label{fig:pruning}
\end{figure}

We can solve the {\AETFS} problem if we can simulate all valid combinations of merge operations represented by $\mathcal{T}$.
However, we do not have to check all combinations explicitly.
This is due to the fact that we can \emph{prune} branches and \emph{merge} nodes in the decision tree $\mathcal{T}$ as follows\footnote{Once the merge operation is applied to the decision tree, it is no longer a tree. However, we continue calling it the decision ``tree'' for convenience.
}.

For incremental $b = 1, 2, \ldots, |\mathcal{I}|-1$,
we edit $\mathcal{T}$ according to the following four rules:
\begin{description}
  \item[Rule 1.] If $M[b] = 0$, then prune all edges to m-nodes at depth $b$.
  \item[Rule 2.] If there is a path $(v_1, v_2, \ldots, v_p)$ such that
    the depth of $v_p$ is $b$, then
    $(v_2, \ldots, v_p)$ is an m-path, and
    $v_1$ and $v_p$ respectively correspond to the length-$k$ prefix and the length-$k$ suffix of the same  sensitive pattern. Hence prune the edge~$(v_{p-1}, v_p)$.
    In other words, we prune the edge $(v_{p-1}, v_p)$ if merging the path strings results in creating a long sensitive pattern.
  \item[Rule 3.] If there are multiple \#-nodes at depth $b$, then merge all of them into a single \#-node.
  \item[Rule 4.]
    If there are multiple m-nodes $\{v_1, v_2, \ldots\}$ at depth $b$ such that
    each $u_i$ does not correspond to the length-$k$ prefix of any sensitive pattern,
    where $u_i$ is the parent of the starting node of the longest m-path ending at $v_i$,
    then merge such m-nodes into a single m-node.
\end{description}
Finally, we add the \emph{sink node} under the decision tree such that
the sink node corresponds to the bottom row $E[|\mathcal{I}|(k+1)][0\dd  n]$, and
each node at depth $|\mathcal{I}|-1$ has only one outgoing edge to the sink node. We also rename the root to the \emph{source node}.

After executing all pruning and merging operations, the decision tree becomes a DAG
whose all source-to-sink paths represent all valid choices~(inspect Figure~\ref{fig:decision} in Appendix~\ref{sec:omitted}).

We call such DAG the \emph{decision DAG}, and we denote it by $\mathcal{G}$.
Although the size of $\mathcal{T}$ can be exponentially large, we can directly construct $\mathcal{G}$ in a top-down fashion
from an instance of {\AETFS} in $\Order(|\mathcal{G}|)$ time.
\paragraph*{Correctness}
We show that no valid path is eliminated and
all invalid paths are eliminated
while constructing $\mathcal{G}$ from $\mathcal{T}$.
Clearly, crossing \#-nodes creates no invalid path.
In what follows, we mainly focus on m-nodes that can create invalid paths.

It is easy to see that
a path $(v_1, v_2, \ldots, v_p)$ is invalid
if and only if
(i) there is a sub-path $(v_s, \ldots, v_t)$ where $v_s$ and $v_t$ respectively correspond to the length-$k$ prefix and suffix of the same sensitive pattern, and $v_{s+1}, \cdots, v_t$ are all m-nodes or
(ii) there is an edge~$(v_{i-1}, v_i)$ such that
$M[d_i] = 0$ where $d_i$ is the depth of $v_i$.
By Rule 1, we delete all invalid paths which satisfy condition (ii),
and do not delete any valid path.
By Rule 2, we delete an invalid path which satisfies condition (i),
and do not delete any valid path.
By Rule 3, we merge \#-nodes, however,
it does not matter since this operation does not cause
deleting or creating any path.
By Rule 4, we may merge m-nodes, however,
the m-nodes to be merged are carefully chosen
to not interfere with Rule 2.
Thus, this also does not cause
deleting or creating any path.
Therefore, $\mathcal{G}$ is constructed correctly.
\paragraph*{The Size of the DAG}
We next analyze the size of $\mathcal{G}$.
Clearly, the total number of \#-nodes is equal to $|\mathcal{I}|-1$.
Also, the source node and the sink node are unique.
The number of m-nodes, each of which is a child of some \#-node, is equal to the number of \#-nodes, i.e., $|\mathcal{I}|-1$.
The number of the rest of m-nodes is at most $L$.
Also, each node has at most two outgoing edges.

Each \#-node and the source node represent a sub-table of size $(k+1)\times (n+1)$. Each m-node represents a possibly mergeable row, and the sink node represents the last gadget row.
Therefore, the total size of $\mathcal{G}$ is $\Order(|\mathcal{I}|kn + |\mathcal{I}|n+Ln) = \Order(k|\mathcal{I}|n + Ln)$.
\paragraph*{Time Complexity}
The decision DAG $\mathcal{G}$ is computed in $\Order(k|\mathcal{I}| + L)$ time without creating the original decision tree $\mathcal{T}$
by applying the above four rules for incremental $b = 1, 2,\ldots, |\mathcal{I}|-1$. Also, we can compute each cell in $\mathcal{G}$ in amortized constant time~\cite{bernardini_et_al:LIPIcs:2020:12132}. Thus, the total time is $\Order(k|\mathcal{I}|n + Ln)$.

\subsection{Short Sensitive Patterns}
Running the {\ETFSDP} algorithm may introduce short sensitive patterns in its gadgets. We explain how to modify the recursive formulae of the gadget row to account for short sensitive patterns.
We first prove that w.l.o.g.~all gadgets are either a single $\#$ or can be optimally aligned such that:
\begin{enumerate}
    \item All \#'s in gadgets are substituted by letters in $W$;
    \item All letters in gadgets are matched with letters in $W$; and
    \item No further letters are inserted between letters of the same gadget.
\end{enumerate}

If some extra inserted letters of $W$ are aligned with a gadget, we can add some extra \#'s to change them into substitutions without increasing the cost, changing the number of non-sensitive patterns or increasing the number of sensitive patterns. Similarly, if some letters of the gadget are not matched with the same letters in $W$, these gadget letters can be replaced by \#. Finally, if some \#'s in the gadget are not aligned with any letter in $W$, we can either remove them or move them to the place of an adjacent gadget letter while deleting that letter. Inspect the following example.
\begin{example}
Let the following optimal alignment from Example~\ref{ex:ETFS} with cost $4$. Gadgets are in red.

\noindent\texttt{e 	c 	a 	b - a 	a 	a 	a 	\textcolor{blue}{a} 	b 	b 	-  	b 	a 	d 	f}\\
\texttt{e 	c 	a 	b 	\textcolor{red}{\# 	a 	a 	\#} 	a 	\textcolor{blue}{b} 	b 	b 	\textcolor{red}{\#} 	b 	a 	d 	f}

We transform it to another optimal alignment of the same cost that respects the above conditions:

\noindent\texttt{e 	c 	a 	b  a 	a 	a 	a 	\textcolor{blue}{a} 	b 	b 	-  	b 	a 	d 	f}\\
\texttt{e 	c 	a 	b 	\textcolor{red}{\# 	a \#} 	a 	\textcolor{blue}{b} 	b 	b 	\textcolor{red}{\#} 	b 	a 	d 	f}

\noindent Let us first consider the leftmost gadget: (1) \#'s are substituted by letters in $W$; (2) all letters are matched with letters in $W$; and (3) no further letters are inserted between the gadget's letters. Note that the rightmost gadget is a single $\#$ and so the modified alignment satisfies all conditions above.
\end{example}
%

A single \# cannot introduce any sensitive pattern, so just as in the {\ETFSDP} algorithm we can get a cost of $E[i-1][j] + 1$ corresponding to the case that a single \# is inserted after $W[j-1]$ or a cost of $E[i-1][j-1] + 1$ corresponding to the case that a single \# is aligned with $W[j-1]$. For longer gadgets the possibilities are a bit more restricted than in the {\ETFSDP} algorithm. Assuming the gadget to have the structure described above, it follows that the previous \# cannot be aligned before $W[F[j-1]]$, where $F[j]$ is defined to be the largest integer such that $W[F[j]\dd j-1]$ contains a sensitive or non-sensitive pattern (if it exists). More formally:
$$F[j] = \max(\{i < j \mid W[i\dd j-1] \text{ contains a sensitive pattern}\} \cup \{j-k\} \cup \{0\}).$$
$F$ can be computed in $\Order(kn)$ time. We denote the point-wise minimum of the copies of the preceding merge row with $r$; in the case of {\ETFS} this is just the previous merge row. This gives us the following formula for the gadget rows. For all $0 \le i \le (k+1)|\mathcal{I}|$ with $i \equiv 0 \mod k+1$,
\begin{eqnarray*}
\label{eq:gadgetrow}
E[i][j] &=& \min\left. \begin{cases}
r[j]+1 \\
r[j-1]+1 \\
E[i][j-1]+1,E[i][j-2]+1,\ldots, E[i][F[j-1]+1]+1.\\
\end{cases}\right.
\end{eqnarray*}
(Notice that a string position and its corresponding table index differ by one.)

To conclude, we also need to consider the range in which the edit distance value lies. Since the last row corresponds to $\otimes$, the value stored in $E[|\mathcal{I}|(k + 1)][j]$, for all $0 \le j \le n$, is the cost of an optimal alignment between $W[0\dd j+e_j-1]$ and a string in the regular expression whose length-$(e_j+1)$ suffix is $\#W[j\dd j+e_j-1]$, where $e_j = \min(\max\{e \mid W[j\dd j+e-1] \text{ does not contain any sensitive or non-sensitive pattern}\}\cup\{n-j\})$.
The edit distance between $W$ and any string matching the regular expression is found among the rightmost $n-F[n]$ entries of the last row or the rightmost entry of the second-last row. Thus, we obtain:
\begin{eqnarray*}
  d_\mathrm{E}(Y_{\mathrm{ED}}, W) &=& \min\left. \begin{cases} 
  E[|\mathcal{I}|(k+1)-1][n], \\
  E[|\mathcal{I}|(k+1)][n], E[|\mathcal{I}|(k+1)][n-1], \ldots, E[|\mathcal{I}|(k+1)][F[n]+1].
  \end{cases}\right.
\end{eqnarray*}


For each $E[i][j]$ and $r[j]$ we store a pointer to an entry which led to this minimum value. We can then trace back as in {\ETFSDP}, taking the minimizing entry of the above equation as a starting point, and obtain $Y_{\mathrm{ED}}$ in an additional $\Order(kn)$ time. Therefore the total time complexity of {\AETFS} is $\Order(k|\mathcal{I}|n+Ln)$ and we arrive at Lemma~\ref{lem:AETFS}. 

\section{$\tilde{\Order}(n^2)$-Time Algorithms using Dyadic Intervals} \label{sec:nsquared}
In this section we improve {\ETFSDP} and the algorithm of Lemma~\ref{lem:AETFS}.
We first show an algorithm to compute {\em gadget rows} in amortized constant time per cell in the rows. Secondly, we focus on the redundancy in the computation of {\em ordinary rows}, and propose an algorithm to compute them by using the lookup table structure of Section~\ref{sec:tab} according to dyadic intervals. These two improvements yield an $\TOrder(n^2)$-time algorithm for {\ETFS}. Finally, we employ a similar lookup table technique to contract m-paths in the decision DAG, which yields an $\TOrder(n^2)$-time algorithm for {\AETFS} as well.
\subsection{Speeding Up Gadget Rows Computation}\label{subsec:gadgetspeedup}
First, we show how to speedup the gadget rows computation.
For each \#-node $u$ in the decision DAG $\mathcal{G}$,
we denote by $d_u$ the in-degree of $u$.
Let $G_u[0\dd  n]$ be the gadget row in $u$.
For each $0 \le i \le d_u-1$, let $M_u^i[0\dd  n]$ be a possibly mergeable row of a node
which has an edge pointing to $u$.
The recursive formula for $G_u[i]$ is as follows:
$G_u[0] = \min\{M_u^0[0] + 1, \ldots, M_u^{d-1}[0] + 1\}$, and
\begin{equation*}
G_u[j] = \min\left.
\begin{cases}
  M_u^0[j-1] + 1, \ldots, M_u^{d-1}[j-1] + 1,\\
  M_u^0[j] + 1, \ldots, M_u^{d-1}[j] + 1,\\
  G_u[j-1] + 1, \ldots, G_u[F[j-1]+1] + 1,\\
\end{cases}\right.
\end{equation*}
for $1 \le j \le n$.
We assume that $M_u^0, \ldots, M_u^{d-1}$ and $F$ are given.
It costs $\Order(n(k+d_u))$ time to compute $G_u$  na\"{\i}vely. The next lemma states that we can actually compute $G_u$ in $\Order(nd_u)$ time.
\begin{lemma} \label{lem:gadgetrow}
Given $M_u^0, \ldots, M_u^{d-1}$ and $F$, we can compute every $G_u[i]$ in $\Order(d_u)$ time for incremental $i = 0, \ldots, n$.
\end{lemma}
\begin{proof}
Let us fix an arbitrary \#-node $u$ and omit subscripts related to $u$.
Let $r_j$ be the index such that $G[r_j]$ is the rightmost minimum value in the range $G[F[j-1]+1\dd  j]$,
and let $m_j = G[r_j]$ be that minimum value.
Then, it can be seen that $G[p] = m_j+1$, for any $r_j < p \le j$,
since $G[p] > G[r_j]$ and $G[p] \le G[r_j] + 1$ by the recursive formula.
Clearly, $r_0 = 0$.
We assume that $r_{j-1}$ is known before computing $G[j]$.
If $r_{j-1} < F[j-1]+1$, then $G[j-1] = m_{j-1} + 1$ is the minimum in $G[F[j-1]+1\dd j-1]$,
and $r_j = \arg\min\{G[j-1], G[j]\}$.
Otherwise, $F[r_{j-1}] = m_j$ is the minimum in $G[F[j-1]+1\dd j-1]$,
and $r_j = \arg\min\{G[r_{j-1}], G[j]\}$.
Note that $F$ is a non-decreasing array, i.e., $F[j] \ge F[j-1]$.
Thus, we can compute $G[j]$ and $r_j$ in $\Order(d)$ time.
\end{proof}
By Lemma~\ref{lem:gadgetrow}, we can compute all gadget rows in a total of $\Order(n\sum_{u\in \mathcal{G}}d_u) = \Order(n|\mathcal{I}| + nL)$ time.
\subsection{{\ETFS} in $\Order(n^2\log^2 k)$ Time} \label{subsec:etfsnsquared}
In this section we describe an algorithm which solves {\ETFS} in $\Order(n^2\log^2 k)$ time. The key to losing the factor $k$ is the fact that the string $T$ on the left is highly repetitive and only consists of substrings of the length-$n$ string $W$ (interleaved by some letters in $\Gamma$). Therefore we can compute the DP table efficiently using only few precomputed sub-tables as in the Four Russians method~\cite{arlazarov1970economical}.

First, we partition $W$ into substrings of length $2^i$ (or shorter if $2^i \nmid |W|$) for each $i \in \{0, 1, 2, \ldots, \lfloor \log k \rfloor\}$.
This gives a set $\mathcal{A}$ of at most $2n$ different strings. Moreover, note that each length-$k$ pattern in $W$ can be written as the concatenation of at most 2$\lfloor \log k \rfloor + 2$ such strings.

For every pair of strings in $(w_1, w_2) \in \mathcal{A}^2$ with $|w_1| = |w_2|$, we precompute the lookup table for the strings $w_1$ and $w_2$ according to Lemma~\ref{lem:traceback}. Now we can compute the non-merge case of each possibly mergeable row using at most $2\cdot \lceil n/2^i\rceil$ precomputed lookup tables of size $2^i \times 2^i$ (or smaller) for each $i \in \{0, 1, 2, \ldots, \lfloor \log k \rfloor\}$.

\paragraph*{Time Complexity} Precomputing a lookup table for two strings of length up to $2^i$ takes $\Order(2^{2i}i)$ time. In total this gives a precomputation time of 
$$\Order\left(\sum_{i = 0}^{\lfloor\log k\rfloor} \frac{n}{2^i} \cdot \frac{n}{2^i} \cdot 2^{2i}i \right) = \Order(n^2\log^2 k).$$
Each possibly mergeable row can now be computed in $\Order(n\log k)$ time from the previous merge and gadget row, since each non-sensitive length-$k$ pattern can be partitioned into at most $2\lfloor \log k \rfloor + 2$ precomputed strings. Gadget rows can be computed in $\Order(n)$ time each from the preceding possibly mergeable rows using the technique described by Lemma~\ref{lem:gadgetrow}.

For the traceback, note that $|X_{\mathrm{ED}}| = \Order(kn)$, i.e., the length of an optimal alignment path over $E$ is $\Order(kn)$. We do not know how the path behaves inside each block. However we can compute the sub-path inside a block in time linear in the path's length by using Lemma~\ref{lem:traceback}. The gadget rows can be traced back in a further $\Order(n)$ time.
Thus, we can trace back in a total time of $\Order(kn)$. 
Therefore the total time complexity is $\Order(n^2\log^2 k)$. We arrive at the following result.


\etfsatheorem*


\subsection{{\AETFS}~in $\Order(n^2\log^2 n)$ Time}
In this section we describe how to further reduce the decision DAG $\mathcal{G}$ from Section~\ref{sec:AETFS} by precomputing parallel m-paths.
First, we give some observations for m-paths.
An m-path is said to be \emph{maximal} if the m-path cannot be extended either forward or backward.
Any two maximal m-paths do not share any nodes, since every m-node in $\mathcal{G}$ has at most one outgoing edge to m-nodes and at most one incoming edge from m-nodes.
Also, the number of maximal m-paths is at most $|\mathcal{I}|$ since the parent of the first m-node of each maximal m-path is a different \#-node or the source node.
In what follows, suppose $\mathcal{G}$ contains a total of $p$ maximal m-paths of length $\ell_1,\ldots,\ell_p$ with $\sum_{i=1}^p \ell_i \le n + \ell n$, where $\ell$ is the length of the longest sensitive pattern. 
Recall that an m-node represents a possibly mergeable row of size $1\times(n+1)$, and thus, we will identify an m-path of length $x$ with a DP table of size $x\times(n+1)$.

Let us now describe our DAG reduction.
An example of the DAG reduction is demonstrated in Figure~\ref{fig:copy}.
In the prepocessing phase, we first construct a lookup table structure for all possible m-paths corresponding to dyadic intervals of lengths at most $\ell$ over the range $[1, |\mathcal{I}|-1]$ of the depths of m-paths, in a similar way as in Section~\ref{subsec:etfsnsquared}.
Next, for each maximal m-path in $\mathcal{G}$, we decompose it into shorter m-paths according to dyadic intervals.
We then contract each such m-path into a single node named \emph{j-node} consisting of a single row, which \emph{jumps} from the beginning of the m-path to the end and represent consecutive merges.
Also, we have to take into account edges leaving the m-path.
Note that paths that leave the m-path early always leave to a \#-node, so we do not have to worry about introducing any sensitive patterns.
We therefore create a new \emph{copy} of the m-path preceded by an additional node named \emph{c-node} consisting of a single row, which takes the point-wise minimum of the parent nodes of the m-paths.

\begin{figure}[t]
    \centering
    \includegraphics[width=\textwidth]{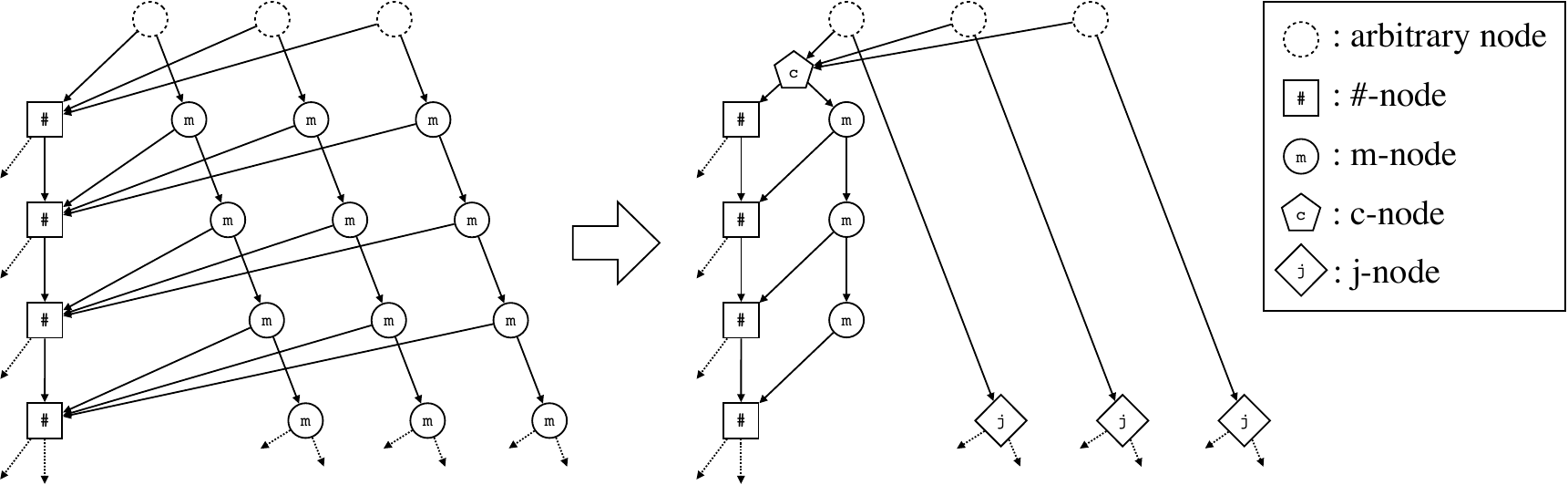}
    \caption{
      An example of contracting a part of the decision DAG.
      Before contracting, there are three parallel m-paths each of length $2^2 = 4$.
      We contract each m-path to a j-node corresponding to the case in which four consecutive non-sensitive patterns are merged.
      Also, we create an m-path of length 3 preceded by a c-node corresponding to the case in which at least one gadget is inserted.
    }
    \label{fig:copy}
\end{figure}

After finishing the DAG reduction, we fill the DP tables from top to bottom: all j-nodes are computed by using the lookup table; all new m-paths are computed in the original fashion, including all outgoing edges; and all the other nodes are computed as in Section~\ref{subsec:etfsnsquared}.
Also, we can trace back and find the solution to {\AETFS} by storing appropriate backward-pointers and the data structures of Lemma~\ref{lem:traceback} just as in Section~\ref{subsec:etfsnsquared}.

\paragraph*{Time Complexity}
Constructing the lookup tables takes $\Order(n^2\log^2 \ell)$ time since there are at most $\lfloor \log \ell\rfloor$ different path lengths, and for each $i\in\{0, 1, \ldots, \lfloor \log \ell\rfloor\}$, we preprocess at most $(n / 2^i)^2$ blocks of size $2^i \times 2^i$ each in $\Order(2^{2i}i)$ time.
We can also easily contract $\mathcal{G}$ in $\Order(n^2)$ time by traversing the DAG. Note that the number of nodes in the original DAG is $\Order(n^2)$~(Section~\ref{sec:longsensitive}).
We partition each path of length $\ell_i$ into at most $2(\log \ell + \ell_i/\ell)$ precomputed paths: at most $\ell_i / 2^{\lfloor \log \ell \rfloor}$ paths of length $2^{\lfloor \log \ell \rfloor}$ and at most 2 of each shorter length.
Therefore the j-nodes can be computed in $\Order(n \cdot \sum_{i=1}^p 2(\log \ell + \ell_i/\ell)) = \Order(n^2 \log \ell)$ time. The c-nodes and the following m-nodes can be computed in $\Order(n^2\log \ell)$ time, because there is at most one c- or m-node per depth and per precomputed path length. The \#-nodes can each be computed in $\Order(n\log^2 k)$ time using the method described in Section~\ref{subsec:etfsnsquared}.
Finally, tracing back takes only $\Order(kn)$ time by using the backward-pointers and the data structures of Lemma~\ref{lem:traceback}.

Summarizing this section, we have shown that the {\AETFS} problem can be solved in time $\Order(n^2\log^2 k + \min\{n^2\log^2\ell, Ln\})$. We arrive at the following result.

\aetfsatheorem*

We defer investigating the generalization of this result for arbitrary weights to the full version of this paper.

\paragraph*{Acknowledgments} We wish to thank Grigorios Loukides (King's College London) for useful discussions about improving the presentation of this manuscript.

\bibliography{bibliography.bib}

\clearpage
\appendix
\section{Omitted Figure} \label{sec:omitted}
\begin{figure}[!h]
\centering
\includegraphics[width=\linewidth]{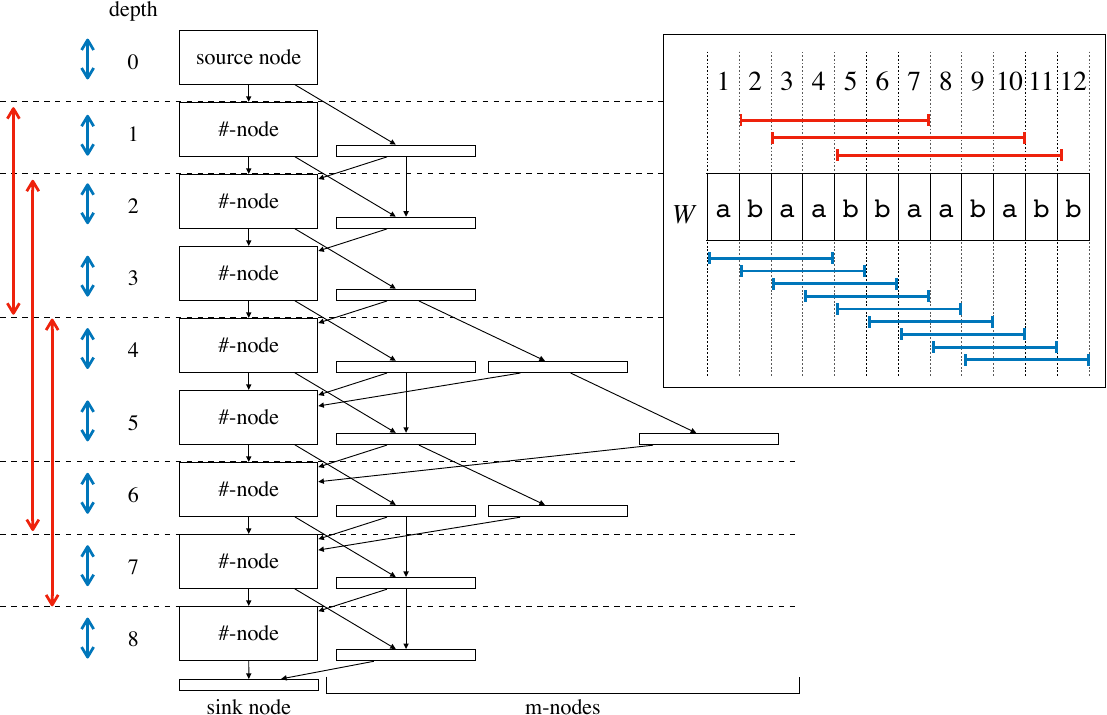}
\caption{
  The decision DAG for
  $W = \mathtt{abaabbaababb}$, $k = 4$, and $\mathcal{S} = \{[2,7], [3,10], [5,11]\}$.
}
\label{fig:decision}
\end{figure}
\end{document}